\documentclass[12pt,a4paper]{article}
\usepackage{amsmath,amsfonts,amssymb,amsthm,mathtools}
\usepackage{authblk}
\usepackage{latexsym,graphicx}
\usepackage{dsfont}
\usepackage{amscd}
\usepackage{mathrsfs}
\usepackage[multiple]{footmisc}

\newcommand{\PP}{\mathsf{P}}

\newcommand{\ee}{{\rm e}}
\newcommand{\ii}{{\rm i}}

\newcommand{\half}{\mbox{$\frac{1}{2}$}}

\newcommand{\R}{{\mathbb R}}
\newcommand{\C}{{\mathbb C}}
\newcommand{\Z}{{\mathbb Z}}

\newcommand{\cP}{{\mathcal  P}}
\newcommand{\cK}{{\mathcal  K}}
\newcommand{\cN}{{\mathcal  N}}

\newtheorem{theorem}{Theorem}[section]
\newtheorem{lemma}[theorem]{Lemma}
\newtheorem{proposition}[theorem]{Proposition}
\theoremstyle{remark}
\newtheorem{remark}[theorem]{Remark}

\title{\Large{\bf{Exact solutions by integrals of the non-stationary elliptic Calogero-Sutherland equation}}}

\date{\vspace{-0.5cm}\small\today\vspace{-0.5cm}}
\author[1,*]{Farrokh Atai}
\affil{Department of Mathematics\\
Kobe University\\
Rokko, Kobe 657--8501, Japan}
\author[2,\dag]{Edwin Langmann}
\affil{Department of Physics\\
KTH Royal Institute of Technology\\
SE-106 91 Stockholm, Sweden}\vspace{2mm}

\date{\vspace{-1.0cm}\small \today}

\begin{document}
\maketitle

\let\oldthefootnote\thefootnote
\renewcommand{\thefootnote}{\fnsymbol{footnote}}
\footnotetext[1]{Electronic address: {\tt farrokh@math.kobe-u.ac.jp}}
\footnotetext[2]{Electronic address: {\tt langmann@kth.se}}
\let\thefootnote\oldthefootnote

\begin{abstract}
We use generalized kernel functions to construct explicit solutions by integrals of the non-stationary Schr\"odinger equation for the Hamiltonian 
of the elliptic Calogero-Sutherland model (also known as  elliptic Knizhnik-Zamolodchikov-Bernard equation). 
Our solutions provide integral represenations of elliptic generalizations of the Jack polyomials. 
\medskip

\noindent MSC class: 33E30, 32A26, 81Q05, 16R60

\medskip

\end{abstract}

\section{Introduction}
\label{sec:Intro}
In this paper we construct explicit solutions of the partial differential equation
\begin{equation} 
\label{Eq:Main}
\left(  \kappa \frac{\ii}{\pi} \frac{\partial}{\partial\tau}  -\sum_{j=1}^{n}\frac{\partial^2}{\partial x_j^2} + \sum_{j<k}^{n} 2g(g-1)\wp(x_j- x_k |\pi,\pi\tau)\right) \psi(x_1,\ldots,x_n;\tau) = 0
\end{equation} 
where $\wp(x|\pi,\pi\tau)$ is the usual Weierstrass function with periods $(2\pi,2\pi\tau)$, $x_j\in[-\pi,\pi]$ are variables on a circle ($j=1,\ldots,n$), 
$g>1/2$, and $\kappa = kg$ with integer $k\geq 1$ (the definition of the function $\wp$ can be found in Appendix~\ref{app:elliptic}; 
our notation otherwise is explained at the end of this section). 
We refer to this as {\em non-stationary elliptic Calogero-Sutherland (eCS) equation} since, for $-\ii\tau>0$ and $\ii\kappa>0$, 
it can be interpreted as non-stationary Schr\"odinger equation for the Hamiltonian defining the eCS model 
\cite{Cal71,Sut72,OP77}\footnote{{\em Calogero-Sutherland model} is short for {\em $A_{n-1}$ quantum Calogero-Moser-Sutherland model}.} 
with the parameter $\ii\kappa$ corresponding to a time scale.\footnote{In physics units $\hbar=2m=1$, the time variable in this Schr\"odinger equation is $t=-\pi\tau/\kappa>0$.}

Our solutions have the form 
\begin{equation*}
\psi(x_1,\ldots,x_n;\tau) =  C(\tau)\left( \prod^{n}_{j\neq k}  \theta\left(z_j/z_k;p\right) \right)^{g/2} \cP(z_1,\ldots,z_n;p) \quad (z_j =\ee^{\ii x_j},\; p=\ee^{2\pi\ii\tau}), 
\end{equation*} 
with   
\begin{equation} 
\label{theta} 
\theta(z;p) \equiv (1-z)\prod_{\ell=1}^\infty (1-p^\ell z)(1-p^\ell/z)\quad (z\in\C,\; |p|<1)
\end{equation} 
the usual multiplicative theta function, and $\cP(z_1,\ldots,z_n;p)$ are symmetric functions in the variables $(z_1,\ldots,z_n)$ 
that reduce to symmetric Laurent polynomials in the limit $p\to 0$.\footnote{For simplicity, we have chosen not to specify the normalization function $C(\tau)$ in the introduction. 
However, this function is important to understand the limit $p\to 0$; see Section~\ref{sec:Result} for details.}
Our main result are explicit integrals representing the latter functions.

In particular, for $\kappa=g$, we construct such solutions $\psi_\lambda(x_1,\ldots,x_n;\tau)$ of \eqref{Eq:Main} labeled by integer vectors $\lambda=(\lambda_1,\ldots,\lambda_n)\in\Z^n$, 
and the corresponding integrals we thus obtain are 
\begin{multline} 
\label{cPintro} 
 \cP_{\lambda}(z_1,\ldots,z_n;p) 
= \left(\prod_{j=1}^n z_{j}^{\lambda_{n}} \right) \left( \prod_{a=1}^{{n}-1} \prod_{j=1}^{a} \oint_{|\xi_{aj}|=\rho_{a}}\frac{d\xi_{aj}}{2\pi\ii\xi_{aj}} \xi_{aj}^{\lambda_{a}-\lambda_{a-1}}\right) \\
\times 
\left(  \frac{\prod_{j\neq k}^{n-1}\theta(\xi_{n-1,j}/\xi_{n-1,k};p)^g}{\prod_{j=1}^{n}\prod_{k=1}^{n-1}\theta(z_{j}/\xi_{n-1,k};p)^g}\right)
 \left(\prod_{a=1}^{n-2}\frac{\prod_{j\neq k}^a\theta(\xi_{aj}/\xi_{ak};p)^g}{\prod_{j=1}^{a+1}\prod_{k=1}^{a}\theta(\xi_{a+1,j}/\xi_{ak};p)^g}\right)
 \end{multline} 
where the integration contours are circles with radii $\rho_a>1$ such that $|p|\rho_{a-1}<\rho_{a}<\rho_{a-1}$ for all $a$, with arbitrary $\rho_{n-1}$ in the range $1<\rho_{n-1}<1/|p|$. 
We show that these integrals define analytic functions in the region $|p|\rho_{n-1}<|z_j|<\rho_{n-1}$ and $|p|<1$. 
Moreover,  when $\lambda_1\geq \lambda_2\geq \cdots\ge \lambda_n$, then the limiting case $p\to 0$ of \eqref{cPintro} reduces (essentially\footnote{There is a difference in the definition of the integration contours; see Remark~\ref{rem:contour} for details. 
Moreover, if $\lambda_n<0$, then we get a Jack polynomial up to a factor $(z_1\cdots z_n)^{\lambda_n}$; see Remark~\ref{rem:Jack}.}) 
to an integral representation of the Jack polynomials \cite{Jac70} obtained by Awata, Matsuo, Odake \& Shiraishi in \cite{AMOS95} in generalization of a result by Mimachi \& Yamada \cite{MY95}, up to explicitly known non-zero constants (see Section~\ref{sec:Result}). 
For integer vectors $\lambda\in\Z^n$ that do {\em not} satisfy the condition $\lambda_1\geq \lambda_2\cdots\geq \lambda_n$, 
the functions $\cP_{\lambda}(z_1,\ldots,z_n;p)$ vanish in the limit $p\to 0$, but they are non-trivial for non-zero $p$. 

Our solution for $\kappa=g$ is complete in the sense that we obtain elliptic generalizations of {\em all} Jack polynomials. 
Our general results are solutions for $\kappa = k g$ for integer $k\geq 1$, but for $k\geq 2$ we only get elliptic generalizations for {\em some} Jack polynomials; 
a precise statement of our results for all integers $k\geq 1$ is given in Theorem~\ref{thm}. 

\begin{remark} 
\label{rem:contour} 
The integral representations of the Jack polynomials obtained from our result in the limit $p\to 0$ differ from the results in \cite{AMOS95} for non-integer $g$ in the definition of the integration contours. 
More specifically, our integration contours are always non-intersecting concentric circles (see Fig.\ 2(a) in \cite{AMOS95}), whereas, 
for non-integer $g$, the ones in \cite{AMOS95} are concentric circles that are pinched so as to intersect at one point (see Fig.\ 2(b) in \cite{AMOS95}). 
To convince critical readers that our simpler integration contours work even for non-integer $g$,  
we give in Appendix~\ref{app:Jack} a concise proof of the result in \cite{AMOS95} explaining this technical point.\footnote{We are grateful to J. Shiraishi for valuable comments prompting us to write Appendix~\ref{app:Jack}.} 
\end{remark} 

We are motivated by Shiraishi's  {\em non-stationary eCS functions} \cite[Section~7]{S19} defined by explicit formal power series and which, 
as Shiraishi conjectured, provide exact eigenfunctions of the non-stationary eCS equation, for arbitrary $\kappa\in\C\setminus\{0\}$ and $g\in\C$.  
As discussed in Section~\ref{sec:Remarks}, we hope that our results can provide a useful test of Shiraishi's conjecture. 

The special case $n=2$ of our results is equivalent to an integral representation of elliptic generalizations of the Gegenbauer polynomials obtained in \cite{LT12}; 
see \cite[Section~3.2]{AL18} in the special case where $\tilde g_\nu=g$ for all $\nu=0,1,2,3$ and \cite{FLNO09} for similar such results. 
The key to our generalization of this result are generalized kernel function identities for the eCS model obtained in \cite{Lan06}.

The non-stationary eCS equation is also known as {\em elliptic Knizhnik-Zamolodchikov-Bernard (KZB) equation} \cite{KZ84,B88} 
in the literature,\footnote{To be more precise: it is a special case of the KZB equations corresponding to $\mathfrak{g} = \mathfrak{sl}_n$ and genus 1; 
see e.g.\ \cite{FV01}.} and integral representations of its solutions were obtained by Etingof \& Kirillov~Jr.\ \cite{EK94a} and Felder \& Varchenko \cite{FV95} using a representation theoretic approach; 
see also \cite{FG97,KT99}. 
These previous results are for $g\in\Z_{\geq 1}$ and $\kappa\in\C\setminus\{0\}$, whereas our results hold true for arbitrary values $g>1/2$ and $\kappa\in g\Z_{\geq 1}$. 
Moreover, our method is different in that it is based on generalized kernel functions rather than representation theory, and our results are more explicit in the sense that they are expressed in terms of more elementary functions and integrations. 

\noindent {\bf Plan:}  In the rest of this paper we review the integral representations of the Jack polynomials (Section~\ref{sec:Jack}), 
give a mathematically precise formulation of our results (Section~\ref{sec:Result}), prove our results (Section~\ref{sec:Proof}), 
and conclude with some remarks to put our results in perspective (Section~\ref{sec:Remarks}). 
For the convenience of the reader we include two appendices with definitions and properties of standard special functions that we use (Appendix~\ref{app:elliptic}) and  
a short proof of the integral representations of the Jack polynomials (Appendix~\ref{app:Jack}). 
A third appendix contains some technical details in our proof in Section~\ref{sec:Proof} (Appendix~\ref{app:Details}).  

\noindent {\bf Notation:} 
Throughout this paper, $n\in\Z_{\geq 1}$ , $\tau\in\C$ with imaginary part $\Im(\tau)>0$,  and $\ii\equiv +\sqrt{-1}$. 
 We write $\sum_{j<k}^{n}$ for the sum over all $j,k=1,2,\ldots,{n}$ such that $j<k$, etc. 
For $\xi$ a complex variable, the symbol $\oint_{|\xi|=\rho} d\xi(\cdots)$ is used for the integral over the counterclockwise oriented circle $|\xi|=\rho$ of radius $\rho>0$. 
We often follow the notation and conventions in Macdonald's book \cite{Mac95}.
In particular, we write $x$ short for $(x_1,\ldots,x_n)$ in the following, and similarly for $z$, etc. 
If $z=(z_1,\ldots,z_n)\in\C^n$ then $1/z$ is short for $(1/z_1,\ldots,1/z_n)$. 
For $L\in\Z_{\geq 1}$, $r=(r_1,\ldots,r_L)\in\Z^L$ and $s=(s_1,\ldots,s_L)\in\Z_{\geq 1}^L$, we sometimes write
\begin{equation} 
\label{shortnotation} 
(r_1^{s_1},\ldots,r_L^{s_L})\; \text{ short for }\; (\underbrace{r_1,\ldots,r_1}_{s_1\text{ times}},\ldots,\underbrace{r_L,\ldots,r_L}_{s_L\text{ times}}),  
\end{equation} 
for example, $(3^2,(-1)^3,2^1)$ is short for $(3,3,-1,-1,-1,2)$.
\bigskip

\noindent {\bf Acknowledgements:} We thank M.~Halln\"as, M.~Noumi, H.~Rosengren, B.~Shapiro, and J.~Shiraishi for helpful discussions.
We are grateful to M.~Noumi and J.~Shiraishi for their interest, encouragement, and advise. 
We gratefully acknowledge partial financial support by the Stiftelse Olle Engkvist Byggm\"astare (contract 184-0573). 
E.L. acknowledges support by VR Grant No. 2016-05167. 
The work of F.A. was carried out as a JSPS International Research Fellow.

\section{Jack polynomials (review)} 
\label{sec:Jack}
For the convenience of non-expert readers, we shortly review a definition and some properties of the Jack polynomials (Section~\ref{sec:DefJack}). 
We also present the integral representations of the Jack polynomials obtained in \cite{AMOS95}, 
elaborating on the integration contours issue described in Remark~\ref{rem:contour} (Section~\ref{sec:IntJack} and Appendix~\ref{app:Jack}). 

The rest of this section can be skipped without loss of continuity.

\subsection{Definition and properties}
\label{sec:DefJack}
For $n\in\Z_{\geq 1}$, we denote by $\PP_n$ the set of all partitions of length $\leq n$, i.e., $\lambda\in\PP_n$ if $\lambda=(\lambda_1,\ldots,\lambda_n)\in\Z^n$ satisfies $\lambda_1\geq \cdots\geq \lambda_n\geq 0$. 
We recall that, for $\lambda\in\PP_n$ and $z=(z_1,\ldots,z_n)\in\C^n$, the {\em monomial symmetric polynomials} are defined as 
\begin{equation*} 
\label{mlambda}
m_{\lambda}(z) \equiv \sum_a z_1^{a_1}\cdots z_n^{a_n}
\end{equation*} 
with the sum over all distinct permutations $a=(a_1,\ldots,a_n)$ of $\lambda=(\lambda_1,\ldots,\lambda_n)$. 

Let $\Lambda_n\equiv \C[z_1,\ldots,z_n]^{S_n}$ be the ring of symmetric polynomials $u(z_1,\cdots,z_n)$ in the variables $z=(z_1,\ldots,z_n)\in\C^n$ 
with complex coefficients.\footnote{We work with complex numbers since we are motivated by quantum mechanics.} 
The monomial symmetric polynomials $m_{\lambda}(z)$ obviously constitute a basis in $\Lambda_n$ labeled by partitions $\lambda\in\PP_n$. 
The Jack polynomials $P^{(1/g)}_{\lambda,n}(z)$ constitute another such basis depending on a parameter $1/g>0$.
They can be defined by a Gram-Schmidt orthogonalization for the following scalar product,  
\begin{equation} 
\label{product} 
\langle f,g\rangle'_{n} \equiv \frac1{n!}\oint_{|z_1|=1}\frac{dz_1}{2\pi\ii z_1}\cdots  \oint_{|z_n|=1}\frac{dz_n}{2\pi\ii z_n}\left(\prod_{j\neq k}^n(1-z_j/z_k) \right)^g f(z)\overline{g(z)}
\quad (f,g\in\Lambda_n)
\end{equation} 
(the bar denotes complex conjugation) using the monomial symmetric polynomials and {\em dominance partial ordering}, 
\begin{equation*} 
\label{leq} 
\mu\leq \lambda\Leftrightarrow \mu_1+\cdots+\mu_j\leq \lambda_1+\cdots+\lambda_j\quad \forall j=1,2,\ldots,n\quad (\mu,\lambda\in\PP_n). 
\end{equation*} 
To be precise, it was proved by Macdonald in \cite{Mac87} that the following two conditions define the Jack polynomials $P^{(1/g)}_{\lambda}\in\Lambda_n$ uniquely, 
\begin{equation*} 
\begin{split}
(a) &\quad  P^{(1/g)}_{\lambda} = m_{\lambda} + \sum_{\mu<\lambda} v_{\lambda\mu}m_{\mu}\; \text{ for some }\; v_{\lambda\mu}\in\R, \\
(b)& \quad \left\langle P^{(1/g)}_{\lambda},m_{\mu}\right\rangle'_{n} = 0\; \text{ for all }\;  \mu<\lambda. 
\end{split}
\end{equation*} 

(In the following we write $P^{(1/g)}_{\lambda,n}$ rather than $P^{(1/g)}_{\lambda}$ to avoid possible confusion.) 

We quote three well-known properties of the Jack polynomials (they are all stated and proved in Macdonald's book \cite{Mac95}). 
First, the following orthogonality relation with respect to the scalar product in \eqref{product}, 
\begin{equation} 
\label{orth}
 \left\langle P^{(1/g)}_{\lambda,n},P^{(1/g)}_{\mu,n}\right\rangle'_{n} = \delta_{\lambda\mu}\cN_{\lambda,n}(g)
\end{equation} 
with the quadratic norms
\begin{equation} 
\label{cN}
\cN_{\lambda,n}(g)\equiv \prod_{j<k}^n\frac{\Gamma(\lambda_j-\lambda_k+g(k-j+1))\Gamma(\lambda_j-\lambda_k+g(k-j+1)+1)}{\Gamma(\lambda_j-\lambda_k+g(k-j))\Gamma(\lambda_j-\lambda_k+g(k-j)+1)}
\end{equation} 
($\delta_{\lambda\mu}$ is the Kronecker symbol equal to 1 for $\lambda=\mu$ and 0 otherwise, and $\Gamma$ is the Euler gamma function). 
Second, a generating function due to Stanley \cite{Sta89},\footnote{Note that, below, we make no distinction between $\lambda\in\PP_m$ and $\lambda=(\lambda_1,\ldots,\lambda_m,0^{n-m})\in\PP_n$, 
i.e., we identify partitions that differ only by a string of zeros.}
\begin{equation} 
\label{GenerateJack}
\prod_{j=1}^n\prod_{k=1}^m\frac1{(1-z_j/\xi_k)^g} = \sum_{\lambda\in\PP_m} \frac1{b_\lambda(g)}P^{(1/g)}_{\lambda,n}(z)P^{(1/g)}_{\lambda,m}(1/\xi)\quad (m\leq n)
\end{equation} 
for $n,m\in\Z_{\geq 1}$, with 
\begin{equation} 
\label{b}
b_\lambda(g)\equiv \prod_{(j,k)\in\lambda} \frac{\lambda_j-k+g(\lambda_k'-j+1)}{\lambda_j-k+1+g(\lambda_k'-j)}\quad (\lambda\in\PP_m)
\end{equation} 
where $\prod_{(j,k)\in\lambda}$ here is short for $\prod_{j=1}^m\prod_{k=1}^{\lambda_j}$ ($\lambda'$ is the partition conjugate to $\lambda$). 
Third, the following Pieri relation of the Jack polynomials, 
\begin{equation} 
\label{shift}
(z_1\cdots z_n)^r P^{(1/g)}_{\lambda,n}(z)= P^{(1/g)}_{\lambda+(r^n),n}(z)\quad (\lambda\in\PP_n, r\in\Z_{\geq 1})
\end{equation} 
with $(r^n)$ short for $(r,\ldots,r)$ ($n$ times). 

\begin{remark}
\label{rem:Jack}
Using \eqref{shift}, one can naturally extend the definition of the Jack polynomials $P^{(1/g)}_{\lambda,n}(z)$ to all {\em ordered integer vectors} $\lambda$, i.e., to all  $\lambda= (\lambda_1,\ldots,\lambda_n)\in\Z^n$ satisfying the condition 
\begin{equation*} 
\label{genlam} 
\lambda_1\geq \cdots\geq \lambda_n 
\end{equation*} 
(that is, one can drop the restriction that $\lambda_n$ is non-negative). 
Thus, for $\lambda_n<0$, we define\footnote{Note that $\lambda-(\lambda_n^n)$ always is a partition under the stated conditions.} 
\begin{equation*} 
P^{(1/g)}_{\lambda,n}(z)\equiv (z_1\cdots z_n)^{\lambda_n}P^{(1/g)}_{\lambda-(\lambda_n^n),n}(z).
\end{equation*} 
In the following we use this generalized definition of the Jack polynomials even though, for $\lambda_n<0$, they are Laurent polynomials. 
Note that the orthogonality relations in \eqref{orth}--\eqref{cN} hold true also for these generalized Jack polynomials, 
and $\cN_{\lambda,n}(g)=\cN_{\lambda+(r^n),n}(g)$ for all $r\in\Z$. 
\end{remark} 

\subsection{Integral representations}
\label{sec:IntJack}
 For $L\in\Z_{\geq 1}$, $r=(r_1,\ldots,r_L)\in\Z^L$ and $s=(s_1,\ldots,s_L)\in\Z_{\geq 1}^L$, let 
 \begin{equation} 
 \label{lambdaa}
\lambda^{(a)}\equiv (r_1^{s_1},\ldots,r_a^{s_a}),\quad N_a\equiv s_1+\cdots+s_a\quad (a=1,\ldots,L), \quad n\equiv N_L, 
\end{equation} 
and 
\begin{subequations}  
\begin{multline} 
\label{cP0thmtrig} 
\cP_{r,s,L}(z_1,\ldots,z_n) \equiv 
\left( \prod_{j=1}^{n} z_j^{r_L}\right) 
 \left(\prod_{a=1}^{L-1}\prod_{j=1}^{N_a}\oint_{|\xi_{aj}|=\rho_a}  \frac{d\xi_{aj}}{2\pi\ii\xi_{aj}} \xi_{aj}^{r_a-r_{a+1}}  \right)
 \\ \times 
 \left(  \frac{\prod_{j\neq k}^{N_{L-1}}(1-\xi_{{N_{L-1},j}}/\xi_{{N_{L-1},k}})^g}{\prod_{j=1}^{n}\prod_{k=1}^{N_{L-1}}(1-z_{j}/\xi_{{N_{L-1},k}})^g}\right)
 \left(\prod_{a=1}^{L-2}\frac{\prod_{j\neq k}^{N_a}(1-\xi_{aj}/\xi_{ak})^g}{\prod_{j=1}^{N_{a+1}}\prod_{k=1}^{N_a}(1-\xi_{a+1,j}/\xi_{ak})^g}\right)
\end{multline} 
with integration contour radii $\rho_a>0$ restricted as follows, 
\begin{equation} 
\label{rhoa0} 
\rho_{a+1}<\rho_a \quad (a=1,\ldots,L-2), \quad 0<|z_j| <\rho_{L-1}\quad (j=1,\ldots,n), 
\end{equation} 
and $\rho_{L-1}>0$ arbitrary. 
\end{subequations}  
 Note that, if and only if $r=(r_1,\ldots,r_L)$ satisfies the condition $r_1\geq \cdots\geq r_L$, 
 then $\lambda^{(a)}-(r_{a+1}^{N_a})$ is a partitions of length $\leq N_a$ for all $a=1,\ldots,L-1$, and therefore 
\begin{equation} 
\label{Cdef}
C_L(r;s)\equiv \prod_{a=1}^{L-1} \frac{N_a!\, \cN_{\lambda^{(a)},N_a}(g)}{b_{\lambda^{(a)}-(r_{a+1}^{N_a})}(g)}
\end{equation} 
is well-defined ($\cN_{\lambda,n}(g)$ and $b_{\lambda}(g)$ are given in \eqref{cN} and \eqref{b}, respectively). 
Moreover, then the Jack polynomial $P^{(1/g)}_{\lambda,n}(z)$ with $\lambda = \lambda^{(L)}$ exists and is non-zero. 

\begin{proposition}[Awata et al.\ \cite{AMOS95}]
\label{thm:AMOS} 
Let $L\in\Z_{\geq 1}$, $r\in\Z^L$, $s\in\Z_{\geq 1}^L$, $\lambda^{(a)}$ as in \eqref{lambdaa}, and $\rho_{L-1}>0$ arbitrary. Then the following hold true. 

\noindent (a) The integrals $\cP_{r,s,L}(z_1,\ldots,z_n)$ in \eqref{cP0thmtrig}--\eqref{rhoa0} are well-defined analytic functions of $(z_1,\ldots, z_n)$ in the region $0<|z_j|<\rho_{L-1}$ $(j=1,\ldots,n)$. 

\noindent (b) If $r=(r_1,\ldots,r_L)\in\Z^L$ satisfies the condition $r_1\geq \cdots\geq r_L$, then 
\begin{equation*} 
\cP_{r,s,L}(z_1,\ldots,z_n) = C_L(r;s) P^{(1/g)}_{\lambda,n}(z_1,\ldots,z_n) \quad (\lambda= \lambda^{(L)}). 
\end{equation*}
Otherwise, $\cP_{r,s,L}(z_1,\ldots,z_n)=0$. 
\end{proposition} 

\begin{proof} 
See Appendix~\ref{app:Jack}. 
\end{proof} 

Every partition $\lambda=(\lambda_1,\ldots,\lambda_n)\in\PP_n$ can be represented as 
\begin{equation} 
\label{lambdareps}
\lambda=(r_1^{s_1},\ldots,r_L^{s_L})
\end{equation} 
with $L=n$, $r=\lambda$, and $s=(1^n)$, and thus $P^{(1/g)}_{\lambda,n}(z)= \cP_{\lambda,(1^n),n}(z_1,\ldots,z_n)/C_n(\lambda;(1^n))$ provides an integral representation for every Jack polynomials. 
If $\lambda$ is such that all $\lambda_j$ are distinct this is the only integral representation of the corresponding Jack polynomial provided by Proposition~\ref{thm:AMOS}. 
However, if some of the $\lambda_j$ are identical, there are several different ways to write $\lambda$ as in \eqref{lambdareps}, 
and thus one gets  from Proposition~\ref{thm:AMOS} several different integral representations  $P^{(1/g)}_{\lambda,n}(z)= \cP_{r,s,L}(z_1,\ldots,z_n)/C_L(r;s)$ of one and the same Jack polynomial. 
The number of integrations in such a representation is 
$$\sum_{a=1}^{L-1}N_a = \sum_{j=1}^{L-1}(L-j)s_j,$$ 
and thus different such representations have different complexity. 

Note that the integral in \eqref{cPintro} is a natural elliptic generalization of the one in \eqref{cP0thmtrig} in the special case $L=n$, $r=\lambda$, and $s=(1^n)$. 

\section{Result}
\label{sec:Result} 
We give a mathematically precise formulation of the main result in this paper. 

We introduce the Euler operator 
\begin{equation} 
\label{PN}
D_{n}(x) \equiv -\ii \sum_{j=1}^{n} \frac{\partial}{\partial x_j}.
\end{equation} 
To simplify some formulas later on we introduce the notation 
\begin{equation} 
\label{wp}
V(x;\tau) \equiv  \wp(x|\pi,\pi\tau)+\frac{\eta_1(\tau)}{\pi}=\sum_{m\in\Z} \frac1{4\sin^2\left(\half x+\pi m \tau\right)}\quad (x\in\C) 
\end{equation} 
and 
\begin{equation} 
\label{eta1} 
 \frac{\eta_1(\tau)}{\pi} \equiv  \frac1{12}+\sum_{m=1}^\infty\frac1{2\sin^2(m\pi\tau)}
\end{equation} 
(see Appendix~\ref{app:elliptic}). Moreover, we find it convenient to write the non-stationary eCS equation as 
\begin{equation} 
\label{Eq:Main1}
\left(  \kappa \frac{\ii}{\pi} \frac{\partial}{\partial\tau}  -\sum_{j=1}^{n}\frac{\partial^2}{\partial x_j^2} + \sum_{j<k}^{n} 2g(g-1)V(x_j- x_k ;\tau)  \right) \psi(x;\tau) = E(\tau) \psi(x;\tau)
\end{equation} 
where $E(\tau)$ is an eigenvalue that we allow to be non-zero. 

To explain the significance of this eigenvalue, we note that \eqref{Eq:Main1} is invariant under  transformations
\begin{equation*} 
\psi(x;\tau) \to  C(\tau) \psi(x;\tau) ,  \quad E(\tau)\to E(\tau)+ \kappa \frac{\ii}{\pi} \frac{1}{C(\tau)} \frac{\partial C(\tau)}{\partial \tau}
\end{equation*} 
for arbitrary differentiable non-zero functions $C(\tau)$. 
Thus, for non-zero $\kappa$,  one can transform $E(\tau)$ to zero, or any other value, by such a transformation. 
In particular, the differential equation in \eqref{Eq:Main1} can be transformed to \eqref{Eq:Main} by a change of the normalization of $\psi(x;\tau)$. 
However, we find it convenient to fix the normalization of $ \psi(x;\tau)$ so that the limits $\kappa\to 0$ and $\Im(\tau)\to+\infty$ are well-defined and non-trivial. 
We thus have to work with \eqref{Eq:Main1} rather than with \eqref{Eq:Main}. 
As will be seen, our normalization is such that $E(\tau)$ is equal to the eigenvalues of the trigonometric Calogero-Sutherland model \cite{Sut72} in the limit $\Im(\tau)\to +\infty$.

As mentioned in Section~\ref{sec:Intro}, we obtain explicit solutions of the non-stationary eCS equation in \eqref{Eq:Main1} 
for $\kappa=k g$ and $k\in\Z_{\geq 1}$ that provide elliptic generalizations of the Jack polynomials $P^{(1/g)}_{\lambda,n}(z)$. 
When $k=1$ then this solution can be constructed for all variable numbers $n\in\Z_{\geq1}$ and $\lambda\in\Z^n$ and is therefore complete in this sense, whereas, 
for $k\geq 2$, $n$ and $\lambda$ are restricted as follows,
\begin{subequations} 
\begin{equation*} 
n=1+k (L-1),\quad \lambda = (r_1,\underbrace{r_2,\ldots,r_2}_{\text{$k$ times}},\ldots,\underbrace{r_L,\ldots,r_L}_{\text{$k$ times}}) 
\end{equation*} 
or 
\begin{equation*}
n=k L,\quad \lambda = (\underbrace{r_1,\ldots,r_1}_{\text{$k$ times}},\ldots,\underbrace{r_L,\ldots,r_L}_{\text{$k$ times}}) 
\end{equation*} 
\end{subequations} 
for arbitrary $L\in\Z_{\geq 1}$ and integer vectors $r=(r_1,\ldots,r_L)\in\Z^L$. 
Using the notation in \eqref{shortnotation} we can distinguish these two cases by a parameter $s_1$ equal to $1$ in the first case and $k$ in the second case. 
Then both cases can be summarized as follows,
\begin{equation} 
\label{lambdagen}
n=s_1+k(L-1),\quad \lambda=(\lambda_1,\ldots,\lambda_n)=(r_1^{s_1},r_2^{k},\ldots,r_L^{k}). 
\end{equation} 
This suggests to label our solutions by the triple $(r,s,L)$ with $s= (s_1,k^{L-1})$. 

The pertinent integrals are given by 
\begin{subequations} 
\begin{multline} 
\label{cPell0} 
\cP_{r,s,L}(z_1,\ldots,z_n;p) \equiv 
\left( \prod_{j=1}^{n} z_j^{r_L}\right) 
 \left(\prod_{a=1}^{L-1}\prod_{j=1}^{N_a}\oint_{|\xi_{aj}|=\rho_a}  \frac{d\xi_{aj}}{2\pi\ii\xi_{aj}} \xi_{aj}^{r_a-r_{a+1}}  \right)
 \\ \times 
 \left(  \frac{\prod_{j\neq k}^{N_{L-1}}\theta(\xi_{{N_{L-1},j}}/\xi_{{N_{L-1},k}};p)^g}{\prod_{j=1}^{n}\prod_{k=1}^{N_{L-1}}\theta(z_{j}/\xi_{{N_{L-1},k}};p)^g}\right)
 \left(\prod_{a=1}^{L-2}\frac{\prod_{j\neq k}^{N_a}\theta(\xi_{aj}/\xi_{ak};p)^g}{\prod_{j=1}^{N_{a+1}}\prod_{k=1}^{N_a}\theta(\xi_{a+1,j}/\xi_{ak};p)^g}\right)
\end{multline} 
with $N_a=s_1+(a-1)k$ ($a=1,\ldots,L-1$) and integration radii $\rho_a$ constrained as follows, 
\begin{equation} 
\label{cPell1} 
|p|\rho_{a}<\rho_{a+1}<\rho_{a}\quad (a=1,\ldots,L-2),\quad |p|\rho_{L-1}<|z_j|<\rho_{L-1}
\end{equation} 
\end{subequations} 
for some arbitrary $\rho_{L-1}$ in the range $1<\rho_{L-1}<1/|p|$. 
Recall that the Jack polynomials are well-defined for all ordered integer vectors $\lambda$ (see Remark~\ref{rem:Jack}). 

\begin{theorem} 
\label{thm}
Let $g>1/2$, $k\in\Z_{\geq 1}$, $L\in\Z_{\geq 1}$, $r=(r_1,\ldots,r_L)\in\Z^L$, $s_1=1$ for $k=1$ and $s_1\in\{1,k\}$ for $k\geq 2$, $s = (s_1,k^{L-1})$, and $\rho_{L-1}$ arbitrary in the range $1<\rho_{L-1}<1/|p|$.
Then the following hold true. 

\noindent (a) The integrals $\cP_{r,s,L}(z_1,\ldots,z_n;p)$ in \eqref{cPell0}--\eqref{cPell1} are analytic functions of $(z,p)=(z_1,\ldots,z_n,p)$ in the following region  $\subset\C^{n+1}$, 
\begin{equation*} 
|p|\rho_{L-1}<|z_j|<\rho_{L-1} \quad (j=1,\ldots,n),\quad |p|<1. 
\end{equation*} 

\noindent (b) The functions 
\begin{equation*}
\psi_{r,s,L}(x_1,\ldots,x_n;\tau) \equiv \left( \prod^{n}_{j\neq k}  \theta\left(z_j/z_k;p\right) \right)^{g/2} \cP_{r,s,L}(z_1,\ldots,z_n;p) \quad (z_j =\ee^{\ii x_j},\; p=\ee^{2\pi\ii\tau})
\end{equation*} 
are solutions of the non-stationary eCS equation in  \eqref{Eq:Main1} for $\kappa=k g$ with the following eigenvalue determined by the integer vector $\lambda$ in \eqref{lambdagen}, 
\begin{equation} 
\label{Ethm} 
E_{r,s,L}(\tau) = \sum_{j=1}^{n}\left( \lambda_{j} + \half g ( n + 1 - 2 j) \right)^{2} +g^2n(n-1)\left(\frac{\eta_1(\tau)}{\pi}-\frac1{12} \right). 
\end{equation}
Moreover, 
\begin{equation} 
\label{pthm}
D_n(x) \psi_{r,s,L}(x;\tau) = d_{r,s,L}\psi_{\lambda}(x;\tau) ,\quad d_{r,s,L} = \sum_{j=1}^n \lambda_j. 
\end{equation} 

\noindent (c) If $r_1\geq \cdots\geq r_L$, then the integer vector $\lambda$ in  \eqref{lambdagen} is ordered, 
and the integral $\cP_{r,s,L}(z_1,\ldots,z_n;p)$ in \eqref{cPell0}--\eqref{cPell1} defines an elliptic generalization of the Jack polynomial $P^{(1/g)}_{\lambda,n}(z)$ in the following sense, 
\begin{equation*} 
\lim_{p\to 0} \cP_{r,s,L}(z_1,\ldots,z_n;p) = C_L(r;s) P^{(1/g)}_{\lambda,n}(z_1,\ldots,z_n)
\end{equation*} 
with the constant $C_L(r;s) $ in \eqref{Cdef}. Otherwise, $\lim_{p\to 0} \cP_{r,s,L}(z_1,\ldots,z_n;p) =0$. 
\end{theorem} 

Thus, by varying $\rho_{L-1}$ in the allowed range, the integral in \eqref{cPell0}--\eqref{cPell1} defines a function in the region $|p|<z_j<1/|p|$ ($j=1,\ldots,n$), $|p|<1$. 
It is interesting to note that the following integration radii satisfy the conditions in \eqref{cPell1}, 
\begin{equation*} 
\rho_a = \rho_{L-1}\epsilon^{1-L+a}\quad (a=1,\ldots,L-1), 
\end{equation*} 
for arbitrary $\epsilon$ and $\rho_{L-1}$ in the range $|p|<\epsilon<1$ and $1<\rho_{L-1}<1/|p|$. 

Part~(c) of Theorem~\ref{thm} is a corollary of the integral representation of the Jack polynomials \cite{AMOS95} reviewed in Section~\ref{sec:IntJack}. 
Indeed, $\lim_{p\to 0}\theta(z;p)=(1-z)$ implies that the integrals in \eqref{cPell0}--\eqref{cPell1} reduce to the ones in \eqref{cP0thmtrig}--\eqref{rhoa0} in the limit $p\to 0$, 
and thus Part~(c) is a simple consequence of Part~(b) of Proposition~\ref{thm:AMOS}. 

We therefore only need to prove Parts (a) and  (b) of this theorem. This is done in Section~\ref{sec:Proof}. 

\section{Proof}
\label{sec:Proof}
In this section we prove Parts (a) and (b) of Theorem~\ref{thm} (Part (c) is a corollary of Proposition~\ref{thm:AMOS}). 
After some preliminaries in Section~\ref{sec:defs}, we give an outline in Section~\ref{sec:outline} where we explain our proof strategy, suppressing technical details. 
The detailed proof is given in Section~\ref{sec:Proof1} with some technical parts deferred to Appendix~\ref{app:Details}. 

\subsection{Definitions and preliminary remarks}
\label{sec:defs}
Recall the definitions of $\theta(z;p)$, $V(x;\tau)$, and $\eta_1(\tau)/\pi$ in \eqref{theta}, \eqref{wp}, and \eqref{eta1}, respectively. 

We recall the definition of the Euler operator $D_{n}(x)$ in \eqref{PN}.  
We also introduce the short-hand notations 
\begin{equation} 
\label{PsiN} 
\Psi_{n}(x;\tau) \equiv \left(\prod_{j\neq k}^{n} \theta\left(z_j /z_k;p\right)\right)^{g/2}\quad (z_j=\ee^{\ii x_j},\; p = \ee^{2\pi\ii\tau})
\end{equation} 
and 
\begin{equation} 
\label{HN}
H_{n}(x;\tau) \equiv -\sum_{j=1}^{n}\frac{\partial^2}{\partial x_j^2} + \sum_{j<k}^{n} 2g(g-1)V(x_j- x_k;\tau).
\end{equation} 
With that we can write the non-stationary eCS equation in \eqref{Eq:Main1} short as follows, 
\begin{equation} 
\label{HNxeq}
\left(  \kappa\frac{\ii}{\pi} \frac{\partial}{\partial\tau}  + H_{n}(x;\tau) \ \right) \psi(x;\tau) = E(\tau)\psi(x;\tau). 
\end{equation} 
Moreover, $\Psi_{n}(x;\tau)$ is a common factor in all our solutions, i.e., the function of interest $\cP(z;p)$ is such that $\psi(x;\tau)=\Psi_{n}(x;\tau)\cP(z;p)$.  
 
To motivate our notation we mention in passing that the differential operator  $H_{n}(x;\tau)$ in \eqref{HN} has a natural physical interpretation as 
Hamiltonian of a quantum many-body system describing $n$ particles moving on a circle and interacting with two-body interactions given by the Weierstrass $\wp$-function, 
$\psi(x;\tau)$ is a quantum mechanical wave function, and $D_{n}(x)$ in \eqref{PN} can be interpreted as  momentum operator.

\subsection{Outline of proof of Theorem~\ref{thm}}
\label{sec:outline} 
We explain the strategy we use to construct the solutions in Theorem~\ref{thm}. 

From now on we simplify notation by suppressing the $\tau$-dependence in our equations if there is no danger of confusion, e.g., we write $H_{n}(x)$ short for $H_{n}(x;\tau)$ etc. 

A simple but important observation is the following: since the operators $ \kappa\frac{\ii}{\pi} \frac{\partial}{\partial\tau}  + H_{n}(x)$ and $D_{n}(x)$ obviously commute, 
we can construct our solutions $\psi(x)$ of \eqref{HNxeq} such that they are also eigenfunctions of $D_{n}(x)$.  

The key to our result is a generalized kernel function $K_{NM}(x,y)$ for the elliptic Calogero-Sutherland model obtained in \cite{Lan06}. 
This function is explicitly known (see \eqref{KNM}), and it satisfies the functional identities \cite{Lan06} 
\begin{equation} 
\label{gen1}
\left( D_N(x)+ D_M(y) \right)K_{NM}(x,y)=0,
\end{equation} 
\begin{equation} 
\label{gen} 
\left( (N-M)g\frac{\ii}{\pi}\frac{\partial}{\partial\tau} + H_N(x)-H_M(y)-c_{NM} \right)K_{NM}(x,y)=0
\end{equation} 
with a known constant $c_{NM}$; see Lemma~\ref{lemma:key} for the precise statement. 

We fix $N-M=k\in\Z_{\geq 1}$ and use this generalized kernel function to construct an integral transform that maps a given solution $\psi_M(y)$ of 
\begin{equation} 
\label{HMyeq1}
D_M(y)  \psi_M(y) = d_M \psi_M(y) ,\quad \left(  k g \frac{\ii}{\pi} \frac{\partial}{\partial\tau}  + H_M(y) \ \right) \psi_M(y) = E_M \psi_M(y) 
\end{equation} 
with known eigenvalues $E_M$ and $d_M$, to a solution of 
\begin{equation} 
\label{HMyeq2}
D_N(x)  \psi_N(x) = d_N \psi_N(x) ,\quad \left(  k g \frac{\ii}{\pi} \frac{\partial}{\partial\tau}  + H_N(x)\right) \psi_N(x) = E_N \psi_N(x) 
\end{equation}
with $N=M+k$ and eigenvalues $d_N$ and $E_N$ that can be computed. Since this new solution $\psi_N(x)$ has a larger number of variables, 
we can use such integral transforms to inductively construct solutions with an increasing sequence of variable numbers $N_a=N_1+(a-1)k$, $a=2,\ldots,L$, 
starting out with a variable number $N_1$ which is such that explicit solutions can be constructed by other methods. 

To obtain a variety of solutions it is important have free parameters in these integral transformations. 
We thus note that, if $\psi_N(x)$, $d_N$, and $E_N$ are solutions of the equations in \eqref{HMyeq2}, then 
\begin{equation} 
\label{Qtransform}
\psi^Q_N(x) = \ee^{\ii Q\sum_{j=1}^N x_j}\psi_N(x),\quad d_N^Q=d_N+NQ,\quad E_N^Q=E_N+2Qd_N+NQ^2
\end{equation} 
are solutions of the same equations for arbitrary real $Q$ (this can be checked by straightforward computations). 
This motivates the following ansatz for the integral transform: 
\begin{equation} 
\label{formal} 
(\cK^{QQ'}_{NM}\psi_M)(x) =  C_{NM}\ee^{\ii Q\sum_{j=1}^N x_j} \int \frac{d^My}{(2\pi)^M} K_{NM}(x,y)   \ee^{-\ii Q'\sum_{k=1}^M y_j}\psi_M(y) 
\end{equation} 
with real parameters $Q,Q'$ and $C_{NM}$ a non-zero constant to be determined. 
We stress that \eqref{formal} is formal since we neither specified the integration domain nor the parameters $Q,Q'$, 
and these details are important to get well-defined solutions. 
However, if we ignore such details for now, we can proceed with formal computations using {\eqref{gen1}--\eqref{HMyeq1}} and dropping boundary terms obtained by partial integrations 
used to transfer the action of the differential operator (details on this computation are given in the proof of Lemma~\ref{lem:1} below). 
We thus find that $\psi_N(x)=(\cK^{QQ'}_{NM}\psi_M)(x)$ satisfies \eqref{HMyeq2} with 
\begin{equation} 
\label{dNEN}
 d_N=d_M+NQ-MQ',\quad E_N = E_M +2(Q-Q')d_M +M(Q')^2 -2MQQ' +  NQ^2 + c_{NM};
\end{equation} 
see Lemma~\ref{lem:1} below for the mathematically precise statement. 

Turning now to our inductive scheme, we use two possible base cases $N=N_1$, namely $N_1=1$ and $N_1=k$ (for $k=1$ these two cases coincide).  
As explained later on, solutions of \eqref{HMyeq2} in these cases are given by 
\begin{subequations} 
\begin{equation} 
\label{psiN1=1}
N_1=1:\quad \psi_{N_1}(x_1)=z_1^{r_1},\quad d_{N_1} = r_1,\quad E_{N_1} = r_1^2
\end{equation} 
and 
\begin{equation} 
\label{psiN1=s}
N_1=k:\quad \psi_{N_1}(x)=\Psi_{k}(x)(z_1\cdots z_{k})^{r_1},\quad d_{N_1} = kr_1,\quad E_{N_1} = {k} r_1^2+c_{k 0}
\end{equation} 
\end{subequations} 
for arbitrary integer $r_1$, with $\Psi_{k}(x)$ the function in \eqref{PsiN} for $n=k$ and $c_{k 0}$ the constant in \eqref{gen} for $(N,M)=(k,0)$ (recall that $z_j=\ee^{\ii x_j}$). 
The induction step from $N_a$ to $N_{a+1}=N_a+k$ is given by the integral transform described above,
\begin{equation} 
\label{subs}
\psi_{N_{a}}(x) = \left(\cK^{Q_{a}Q'_{a}}_{N_{a}N_{a-1}}\psi_{N_{a-1}}\right)(x) 
\end{equation} 
for $a=2,3,\ldots, L$. 
We find that, at each such step, {the possible} parameters $Q_{a}$ and $Q_{a}'$ {are} determined by integer parameters $r_a$ and $r_a'$, 
and one can set $r'_a=r_a$, without loss of generality. 
Our solutions $\psi_{r,s,L}(x)$ are obtained as $\psi_{N_{L}}(x)$. 

With the outline of our methods explained, we now turn to making these arguments precise.

\subsection{Proof of Theorem~\ref{thm}} 
\label{sec:Proof1}
We prove Theorem~\ref{thm} following the outline in Section~\ref{sec:outline}.
\label{sec:proof}

\subsubsection{Generalized kernel function identity}
The result in this section was obtained in Ref.\ \cite{Lan06}.

Recall the definitions of $D_{n}(x)$ and $H_{n}(x)$  in \eqref{PN} and \eqref{HN}, respectively. 
We define
\begin{equation} 
\label{vtet}
\vartheta(x) \equiv 
2\sin\left(\half x\right)\prod_{m=1}^\infty\left(1-2p^m\cos(x) + p^{2m}\right) \quad (x\in\C,\; p=\ee^{2\pi\ii\tau})
\end{equation} 
proportional to the Jacobi theta function $\vartheta_1\left(\half x|\tau\right)$ (see Appendix~\ref{app:elliptic}). 
This allows us to write the function in \eqref{PsiN} as
 \begin{equation} 
\label{PsiN1} 
\Psi_{N}(x) = \left(\prod_{j<k}^{N} \vartheta\left(x_j - x_k\right)^2  \right)^{g/2}\quad (x\in\C^N)
\end{equation} 
(this is true since $\vartheta(x_j-x_k)^2=\theta(z_j/z_k;p)\theta(z_k/z_j;p)$ for $|z_j/z_k|=1$; note that we always assume that this latter condition is fulfilled).
With that we can define the generalized kernel function mentioned in the previous section as 
\begin{equation} 
\label{KNM}
K_{NM}(x,y) \equiv \frac{\Psi_N(x)\Psi_M(y)}{\left(\prod_{j=1}^N\prod_{k=1}^M\vartheta\left(x_j-y_k\right)\right)^g}
\end{equation} 
with $x=(x_1,\ldots,x_N)\in\C^N$ and $y\in(y_1,\ldots,y_M)\in\C^M$. 

\begin{lemma}
\label{lemma:key} 
For all $N,M\in \Z_{\geq 0}$ and $g\in\C$, the function $K_{NM}(x,y)$ in \eqref{KNM} satisfies the functional identities in \eqref{gen1}--\eqref{gen} with the constant 
\begin{equation} 
\label{cNM}
c_{NM} = g^2[N(N-1)-M(M-1)]\frac{\eta_1}{\pi} +\frac1{12}g^2(N-M)[N(N-1)+M(M-1)-2NM].
\end{equation} 
\end{lemma} 

\begin{proof} 
This result is stated and proved in Ref.\ \cite{Lan06}; see first part of the Theorem. 
For the convenience of the reader we give a table translating the symbols used here to the ones used in Ref.\ \cite{Lan06}: 

\begin{equation*}
\label{table}
\begin{array}{|l||c|c|c|c|c|c|c|c|}
\hline
\text{Symbols used here}                &  \vartheta(x)                  & g & \tau & p & \frac{\eta_1}{\pi} & H_N(x) & D_N(x) & K_{NM}(x,y)  \\
\hline 
\text{Symbols used in \cite{Lan06}} & 2\theta\left(x\right) &\lambda &  \frac{\ii\beta}{2{\pi}} & q^2 & c_0 & H_{\lambda,N}(\mathbf{x}) & -P_N(\mathbf{x})  & F_{N,M}(\mathbf{x};\mathbf{y}) \\
\hline
\end{array}
\end{equation*} 
(we ignore an irrelevant multiplicative factor between $K_{NM}(x,y)$ and $F_{N,M}(\mathbf{x};\mathbf{y})$ coming from the factor 2 between $ \vartheta(x)$ and $\theta\left(x\right)$). 

A careful reader will notice that the proof in \cite{Lan06} is for the function $\Psi_{N}(x) $ given by 
\begin{equation*} 
\label{PsiN2} 
\left(\prod_{j<k}^{N} \vartheta\left(x_j - x_k\right)  \right)^{g}
\end{equation*} 
differing from the function $\Psi_{N}(x) $ in \eqref{PsiN1} by a phase factor. 
However, this phase factor is locally constant (it only changes if one crosses branch cuts), and it therefore does not affect the proof of Lemma~\ref{lemma:key} in  \cite{Lan06}. 
\end{proof} 

\subsubsection{Integral operator}
To make the formal integral in \eqref{formal} precise, we replace 
\begin{equation} 
\label{dMy} 
\int\frac{d^My}{(2\pi)^M}\to \int_{[-\pi-\ii\epsilon,\pi-\ii\epsilon]^M}\frac{d^My}{(2\pi)^M} \equiv  \int_{-\pi-\ii\epsilon}^{\pi-\ii\epsilon}\frac{dy_1}{2\pi} \cdots \int_{-\pi-\ii\epsilon}^{\pi-\ii\epsilon}\frac{dy_M}{2\pi}
\end{equation} 
with $\epsilon$ a parameter further constrained below. 
This parameter $\epsilon$ allows us to deform the integration paths in the complex plane so as to avoid singularities of the integrand, 
and by this we can make the computations described after \eqref{formal} mathematically precise. 
By straightforward computations we obtain explicit formulas for the transformed function $\psi_N(x)=(\cK^{QQ'}_{NM}\psi_M)(x)$, 
which is well-defined provided $\psi_M(y)$ satisfies certain conditions. 
We also find that it is convenient to choose the constant in \eqref{formal} as
\begin{equation} 
\label{CNM}
C_{NM} = \ee^{\pi\ii gNM/2} .
\end{equation} 
We summarize our result as follows. 

\begin{lemma} 
\label{lem:1}
Let $g>1/2$, $M,N\in\Z_{\geq 1}$ with $N>M$, 
\begin{equation} 
\label{psiM} 
\psi_M(y;\tau) = \Psi_M(y;\tau)\cP_M(\xi;p)\quad (\xi=(\ee^{\ii y_1},\ldots,\ee^{\ii y_M}),\quad p=\ee^{2\pi\ii\tau}) 
\end{equation}
with $\cP_M(\xi;p)$ an analytic function of $(\xi,p)=(\xi_1,\ldots,\xi_M,p)$ in the following region $\subset\C^{M+1}$,   
\begin{equation} 
\rho_0|p|<|\xi_j|<\rho_0 \quad (1\leq j\leq M),\quad |p|<1
\end{equation} 
for some $\rho_0>0$, and 
\begin{equation} 
\label{QR}
Q=r-Mg/2,\quad Q'=r' - Ng/2
\end{equation} 
with $r,r'\in\mathbb{Z}$. 
Then the integral transform $(\cK^{QQ'}_{NM}\psi_M)(x;\tau)$ of $\psi_M(y;\tau)$ in \eqref{formal} becomes well-defined 
by the substitution in {\eqref{dMy}} provided that $\rho\equiv\ee^\epsilon$ is in the range $|p|\rho_0<\rho<\rho_0$. 
Furthermore, 
\begin{equation} 
\label{precise1}
(\cK^{QQ'}_{NM}\psi_M)(x;\tau) = \Psi_N(x;\tau) (\hat{\cK}^{rr'}_{NM}\cP_M)(z;p)  
\end{equation} 
with $z=(\ee^{\ii x_1},\ldots,\ee^{\ii x_N})$ {where}
\begin{equation}
\label{precise2}  
(\hat\cK^{rr'}_{NM}\cP_M)(z;p) =  
\left( \prod_{j=1}^N z_j^r\right) 
\left(\prod_{k=1}^M\oint_{|\xi_k|=\rho}\frac{d\xi_k}{2\pi\ii\xi_k}\xi_k^{-r'} \right)\frac{\prod_{j\neq k}^{M}\theta(\xi_j/\xi_k;p)^g}{\prod_{j=1}^N\prod_{k=1}^M\theta\left(z_j/\xi_k;p \right)^g}\cP_M(\xi;p)
\end{equation} 
is an analytic function of $(z , p) = ( z_1, \ldots, z_N, p)$ in the following region $\subset\C^{N+1}$, 
\begin{equation} 
\label{zjregion}
|p|\rho<|z_j|<\rho \quad (1\leq j\leq N),\quad |p|<1. 
\end{equation}  
Finally, if $\psi_M(y;\tau)$ satisfies the equations in \eqref{HMyeq1} with known eigenvalues $d_M$ and $E_M$, then 
$\psi_N(x;\tau)=(\cK^{QQ'}_{NM}\psi_M)(x;\tau)$ satisfies the equations in  \eqref{HMyeq2} with eigenvalues $d_N$ and $E_N$ given in 
\eqref{dNEN}, \eqref{cNM} and \eqref{QR}.   
\end{lemma} 

\begin{proof} 
Evaluating \eqref{formal}, using \eqref{KNM}, \eqref{dMy} and \eqref{psiM}, we obtain
\begin{equation*} 
(\cK^{QQ'}_{NM}\psi_M)(x) =
C_{NM}\Psi_N(x)  \ee^{\ii Q \sum_{j=1}^N x_j}
\int_{[-\pi-\ii\epsilon,\pi-\ii\epsilon]^M}\frac{d^My}{(2\pi)^M} \frac{\Psi_M(y)^{2} \ee^{-\ii Q'\sum_{j=1}^M y_j}}{{\prod_{j=1}^N\prod_{k=1}^M\vartheta\left(x_j-y_k\right)^g}}\cP_M(\xi). 
\end{equation*} 
We change variables to $z_j=\ee^{\ii x_j}$ and $\xi_k=\ee^{\ii y_k}$ and recall the definition of $\Psi_N(x)$ in \eqref{PsiN}. 
Inserting  $\vartheta(x_j-y_k) = \ee^{\pi\ii/2}\ee^{-\ii(x_j-y_k)/2}\theta(z_j/\xi_k)$, which follow from the definitions in \eqref{theta} and \eqref{vtet} for $|p|<|z_j/\xi_k|<1$,  
we obtain, by straightforward computations, 
\begin{multline*} 
(\cK^{QQ'}_{NM}\psi_M)(x) = \Psi_N(x) (z_1\cdots z_N)^{Q+Mg/2}
C_{NM} \left(\ee^{\pi\ii/2} \right)^{-gNM}
 \\ \times 
 \oint_{|\xi_1|=\rho}\frac{d\xi_1}{2\pi\ii\xi_1}\cdots  \oint_{|\xi_M|=\rho}\frac{d\xi_M}{2\pi\ii\xi_M}
 (\xi_1\cdots \xi_M)^{-(Q'+Ng/2)}
\frac{\prod_{j\neq k}^M \theta(\xi_j/\xi_k)^g}{\prod_{j=1}^N\prod_{k=1}^M\theta(z_j/\xi_k)^g}
\cP_M(\xi)
\end{multline*}
with $\rho=\ee^{\epsilon}$. We insert $C_{NM}$ in \eqref{CNM} and $Q,Q'$ in \eqref{QR} to obtain \eqref{precise1}--\eqref{precise2}.  

The proof that the integral transform in \eqref{precise2} defines a analytic function is deferred to Appendix~\ref{app:Analyticity}. 

We are left to prove that, if $\psi_M(y)$ satisfies \eqref{HMyeq1}, then  $\psi_N(x)$ in \eqref{formal}
satisfies  \eqref{HMyeq2}. To simplify notation, we write in the following $\int$ short for $\int_{[-\pi-\ii\epsilon,\pi-\ii\epsilon]^M}$, 
and we use the short hand notations $|x|=\sum_j x_j$ and similarly for $|y|$. 
We start with
\begin{multline*} 
D_N(x)\psi_N(x) = D_N(x) C_{NM}\ee^{\ii Q|x|}\int \frac{d^My}{(2\pi)^M} K_{NM}(x,y)\ee^{-\ii Q'|y|}\psi_M(y) 
\\
= C_{NM}\ee^{\ii Q|x|}\int \frac{d^My}{(2\pi)^M} \left(D_N(x) +NQ \right) K_{NM}(x,y)\ee^{-\ii Q'|y|}\psi_M(y) \\
= C_{NM}\ee^{\ii Q|x|}\int \frac{d^My}{(2\pi)^M} \left\{ \left(-D_M(y) +NQ \right) K_{NM}(x,y)\right\} \ee^{-\ii Q'|y|}\psi_M(y), 
\end{multline*}
using \eqref{gen1} in the last equality, where the interchange of integration and differentiations in the second equality is 
justified by  the Leibniz integral rule\footnote{See e.g.\ \cite[Section~4.2]{WW} for a precise mathematical formulation of the Leibniz integral rule. 
Note that, in all cases where we use this rule in this paper, the pertinent function in the integrand satisfies analyticity properties that are much stronger than what is needed to prove the validity of this rule.} 
using the properties of the integrand established above (we use $\{\cdots\}$ to indicate that derivatives act only within this bracket, here and in the following). 
We compute this by partial integrations using that the boundary term vanishes, i.e., 
\begin{equation} 
\label{BT11} 
\int \frac{d^My}{(2\pi)^M}\biggl(  \{ D_M(y)K_{NM}(x,y)\} \ee^{-\ii Q'|y|}\psi_M(y) + K_{NM}(x,y) D_M(y)\ee^{-\ii Q'|y|}\psi_M(y)  \biggr) = 0
\end{equation} 
(this is proved in  Appendix~\ref{app:BT}). We thus obtain 
\begin{multline*} 
D_N(x)\psi_N(x) 
= C_{NM}\ee^{\ii Q|x|}\int \frac{d^My}{(2\pi)^M} K_{NM}(x,y)  \left(D_M(y) +NQ \right) \ee^{-\ii Q'|y|}\psi_M(y)\\
= C_{NM}\ee^{\ii Q|x|}\int \frac{d^My}{(2\pi)^M} K_{NM}(x,y)  \ee^{-\ii Q'|y|}\left(D_M(y) +NQ -MQ'\right) \psi_M(y)
\\ = (d_M +NQ-MQ')\psi_N(x) 
\end{multline*}
using $D_N(y)\psi_M(y)=d_M\psi_M(y)$ in the last step. This proves the first identity in \eqref{HMyeq2} with $d_N$ in \eqref{dNEN}.

In a similar manner, 
\begin{multline*} 
\left(k g\frac{\ii}{\pi}\frac{\partial}{\partial\tau} + H_N(x)\right)\psi_N(x) = 
C_{NM}\ee^{\ii Q|x|}\int \frac{d^My}{(2\pi)^M}\biggl(k g\frac{\ii}{\pi}\frac{\partial}{\partial\tau}+ H_N(x) \\ + 2QD_N(x) +NQ^2
\biggr) K_{NM}(x,y) \ee^{-\ii Q'|y|}\psi_M(y) 
 = C_{NM}\ee^{\ii Q|x|}\int \frac{d^My}{(2\pi)^M} \Biggl[ \biggl\{\biggl(H_M(y)+c_{NM} \\ - 2QD_M(y) +NQ^2
\biggr)  K_{NM}(x,y)\biggr\} \ee^{-\ii Q'|y|}\psi_M(y) 
+ K_{NM}(x,y)k g\frac{\ii}{\pi}\frac{\partial}{\partial\tau}\ee^{-\ii Q'|y|}\psi_M(y)  \Biggr] 
\end{multline*}
using the Leibniz integral rule and \eqref{gen1}--\eqref{gen} to get the second identity (recall that $k=N-M$). 
Again, we proceed by partial integrations using that that the boundary terms which arise vanish, i.e., \eqref{BT11} and 
\begin{equation} 
\label{BT22} 
\int \frac{d^My}{(2\pi)^M} \biggl( \left\{ H_M(y) K_{NM}(x,y)\right\}\ee^{-\ii Q'|y|}\psi_M(y) 
- K_{NM}(x,y)H_M(y) \ee^{-\ii Q'|y|}\psi_M(y) \biggr)=0 
\end{equation}  
(this is proved in  Appendix~\ref{app:BT}). Thus, 
\begin{multline*} 
\left(k g\frac{\ii}{\pi}\frac{\partial}{\partial\tau} + H_N(x)\right)\psi_N(x) = 
 C_{NM}\ee^{\ii Q|x|}\int \frac{d^My}{(2\pi)^M}K_{NM}(x,y)
\biggl( k g\frac{\ii}{\pi}\frac{\partial}{\partial\tau} + H_M(y) +c_{NM}\\
+2QD_M(y) + NQ^2  \biggr) \ee^{-\ii Q'|y|}\psi_M(y) = 
 C_{NM}\ee^{\ii Q|x|}\int \frac{d^My}{(2\pi)^M}K_{NM} \ee^{-\ii Q'|y|}
 \biggl( k g\frac{\ii}{\pi}\frac{\partial}{\partial\tau} + H_M(y) \\
-2Q' D_M(y) +M(Q')^2 +c_{NM}+2 Q D_M(y) - 2MQQ' + NQ^2 \biggr)\psi_M(y) \\
= (E_M +2(Q-Q')d_M +M(Q')^2 -2MQQ' +  NQ^2 + c_{NM})\psi_N(x)
\end{multline*}
using \eqref{HMyeq1} in the last step. This proves the second identity in \eqref{HMyeq2} with $E_N$ in \eqref{dNEN}.
\end{proof} 

\subsubsection{Induction}
We are now ready to prove Theorem~\ref{thm} by induction over $L$. 

For $L=1$ we need to prove that 
\begin{equation} 
\label{psiN1}
\begin{split}
\psi_{N_1}(x) &= \Psi_{N_1}(x)\left( \prod_{j=1}^{N_1} z_j^{r_1}\right) ,\quad d_{N_1} = N_1r_1,\\
E_{N_1} &= \sum_{j=1}^{N_1}\bigl( r_1 + \half g ( N_1 + 1 - 2 j) \bigr)^{2} + g^2N_1(N_1-1)\left(\frac{\eta_1}{\pi}-{\frac{1}{12}}\right) 
\end{split}
\end{equation} 
solve \eqref{HMyeq2} for $N=N_1$. 
For $r_1=0$ this follows from Lemma~\ref{lemma:key}  in the special case $N=N_1$, $M=0$ (this is non-trivial only in case $N_1=k>1$). 
The general case of non-zero $r_1$ follows from this by a simple computation explained in Section~\ref{sec:outline}; 
see \eqref{Qtransform} (note that $\prod_{j}z_j^{r_1} = \exp(\ii r_1 \sum_j x_j)$, and that $E_{N_1}$ in \eqref{psiN1=s} and \eqref{psiN1} agree).

For $L>1$ one can construct the function $\psi_{r,s,L}(x_1,\ldots,x_n;\tau)$ in Theorem~\ref{thm} inductively using the integral operator in \eqref{precise1}--\eqref{precise2} as follows: 
$\psi_{\lambda}(x)=\psi_{N_L}(x)$ with 
\begin{equation} 
\label{iterate}
\psi_{N_{a}}(x) = (\cK^{Q_aQ'_a}_{N_aN_{a-1}}\psi_{N_{a-1}})(x), \quad Q_a=r_a-N_{a-1}\frac{g}{2},\quad Q'_a=r_a-N_{a}\frac{g}{2} 
\end{equation} 
for $a=2,3,\ldots,L$, with $\psi_{N_1}(x)$ in \eqref{psiN1}. 
Then the repeated application of Lemma~\ref{lem:1} guarantees that the function $\psi_{N_{a}}(x)$ has the claimed properties.  

We are left to show that $d_{\lambda}=d_{N_L}$ and $E_{\lambda}=E_{N_L}$ where 
\begin{subequations} 
\begin{equation}
\label{ppNa} 
d_{N_a}=d_{N_{a-1}}+ N_aQ_a-N_{a-1}Q_a', 
\end{equation} 
\begin{equation} 
\label{ENa}
E_{N_a} = E_{N_{a-1}} + 2(Q_a-Q'_a)d_{N_{a-1}} +N_{a-1}(Q'_a)^2-2N_{a-1}Q_aQ_a' +  N_aQ_a^2 + c_{N_aN_{a-1}}
\end{equation} 
\end{subequations} 
for $a=2,3,\ldots,L$. This can be verified by straightforward computations (the interested reader can find details in Appendix~\ref{app:computation}).

\section{Concluding remarks}
We give some outlook and conjectures. 

\label{sec:Remarks}
\noindent {\bf 1.} As explained in more detail below, results in \cite{S19,AL18,EK94b,EFK95,Lan07,Lan14} suggest to us that, for arbitrary $n\in\Z_{\geq 1}$, $g\in\C$, $\kappa\in\C\setminus\{0\}$ and $\lambda\in\C^n$, 
there exist functions $P^{(1/g)}_{\lambda;\kappa,n}(z;p)$ such that $\Psi_n(z;p)P^{(1/g)}_{\lambda;\kappa,n}(z;p)$ is a solution of the non-stationary eCS equation in \eqref{Eq:Main1} with eigenvalues as on the RHS in \eqref{Ethm}, and
\begin{equation*} 
\sum_{j=1}^n z_j\frac{\partial}{\partial z_j} P^{(1/g)}_{\lambda;\kappa,n}(z;p) = |\lambda|P^{(1/g)}_{\lambda;\kappa,n}(z;p),\quad |\lambda|\equiv \sum_{j=1}^n \lambda_j . 
\end{equation*} 
These functions are analytic functions of $(z,p)$ and $(\lambda,\kappa)$ in suitable regions $\subset \C^{n+1}$ and, for partitions $\lambda$, they are symmetric functions in $z$ that reduce to  the Jack polynomials $P^{(1/g)}_{\lambda,n}(z)$ in the limit $p\to 0$. 
We believe that there is a beautiful mathematical theory of this two-parameter deformation of the Jack polynomials which, at this point, is only partly developed. 
The result in the present paper is a simple special case or this.

We now discuss in more detail the reasons for these expectations. 
\begin{itemize} 
\item[$(i)$] It seems that the {\em affine analogue of the Jack polynomials} defined and studied in \cite{EK94b,EFK95} are special cases of these functions $P^{(1/g)}_{\lambda;\kappa,n}(z;p)$.
\item[$(ii)$] We constructed series representations of such functions $P^{(1/g)}_{\lambda;\kappa,n}(z;p)$ for the simplest non-trivial case $n=2$ in \cite{AL18}.\footnote{The results in \cite{AL18} are actually for the non-stationary Heun equation, but this includes as special case the non-stationary Lam\'e equation which, essentially, is the non-stationary eCS equation for $n=2$.}  
\item[$(iii)$]  A construction of singular series solutions (in the sense of \cite{Lan07}) of the non-stationary eCS equation for $\kappa=0$ and arbitrary $n\in\Z_{\geq 2}$ by one of us in \cite{Lan14}, which is straightforward to extend this construction to $\kappa\neq 0$. In fact, this construction is easier for non-zero $\kappa$ \cite{AL18}: first, if $\kappa\neq 0$, one need not compute the eigenvalue $E(\tau)$ in \eqref{Eq:Main1} and, second, if $\Im(\kappa)\neq 0$, it is easy to prove absolute convergence of these singular series solutions in a suitable domain $\subset \C^{n+1}$ using the method in \cite{Lan14}.
 Moreover, using the kernel functions in \eqref{KNM}, one can transform these singular solutions to such described in the previous paragraph (this is explained in \cite{Lan07}), but the latter is restricted to $\kappa=k g$ with $k\in\Z$ and, for $k>0$, these solutions are incomplete (for $k<0$ they are over-complete). 
\item[$(iv)$] Our main reason are Shiraishi's non-stationary eCS functions $f^{\mathrm{eCS}}(z;p|\lambda;\kappa|g)$ in \cite[Section~7]{S19} mentioned already in the introduction. 
 As Shiraishi conjectured, these functions are such that\footnote{Note that our $z$, $\kappa$ and $g$ correspond to $x$, $-k$ and $\beta$ in \cite[Section~7]{S19}, respectively.} 
  \begin{equation} 
 \label{eCSfunction} 
 z_1^{\lambda_1}\cdots  z_n^{\lambda_n}  f^{\mathrm{eCS}}(z;p|\lambda;-\kappa|g)
 \end{equation} 
have exactly the properties of the functions $P^{(1/g)}_{\lambda;\kappa,n}(z;p)$ described in the first paragraph of this section. 
\end{itemize} 

\noindent {\bf 2.} It would be interesting to prove that the elliptic generalizations of the Jack polynomials constructed in the present paper  are identical with the functions in \eqref{eCSfunction} in special cases. 
Since both this functions are given by explicit formulas, this non-trivial test of Shiraishi's conjecture in \cite[Section~7]{S19} seems feasible. 

\noindent {\bf 3.} It is interesting to speculate what the generalizations of the properties in \eqref{orth}--\eqref{shift} of the Jack polynomials to the functions $P^{(1/g)}_{\lambda;\kappa,n}(z;p)$ described in Remark~{\bf 1} above are. We believe that the result in the present paper suggest the following (we stress that the statements in the following three equations are conjectures). 

First, the natural elliptic generalization of the scalar product in \eqref{product} is 
\begin{equation*} 
\label{product2} 
\langle f,g\rangle'_{n} \equiv \frac1{n!}\oint_{|z_1|=1}\frac{dz_1}{2\pi\ii z_1}\cdots  \oint_{|z_n|=1}\frac{dz_n}{2\pi\ii z_n}\left(\prod_{j\neq k}^n\theta(1-z_j/z_k;p) \right)^g f(z;p)\overline{g(z;\bar{p})}, 
\end{equation*} 
and we expect that the natural generalization of the orthogonality relations in \eqref{orth}  is 
\begin{equation*} 
\label{orth2}
 \left\langle P^{(1/g)}_{\lambda;\kappa,,n},P^{(1/g)}_{\mu;-\bar{\kappa},n}\right\rangle'_{n} = \delta_{\lambda\mu}\cN_{\lambda;\kappa,n}(p,g)
\end{equation*} 
for some $\cN_{\lambda;\kappa,n}(p,g)$ that may depend on $p$ and $\kappa$ (this is suggested to us by the quantum mechanical interpretation of the model, using some quantum mechanics folklore). 
Second, we expect the following generalization of the generating function in \eqref{GenerateJack} holds true, 
\begin{equation*} 
\label{GenerateJack2}
\prod_{j=1}^n\prod_{k=1}^m\frac1{\theta(z_j/\xi_k;p)^g} = \sum_{\lambda} \frac1{b_{\lambda;n,m}(p,g)}P^{(1/g)}_{\lambda;(n-m)g,n}(z;p)P^{(1/g)}_{\lambda;(m-n)g,m}(1/\xi;p)\quad (m\leq n)
\end{equation*} 
where the sum is over all integer vectors $\lambda\in{\Z^m}$, and the constants $b_{\lambda;n,m}(p,g)$ may depend on $n$, $m$ and $p$ 
(this, it seems to us, is the unique equation consistent with all known facts). 
Third, it is natural to expect that the Pieri relation in \eqref{shift} remains unchanged,  
\begin{equation*} 
\label{shift2}
(z_1\cdots z_n)^r P^{(1/g)}_{\lambda;\kappa,n}(z;p)= P^{(1/g)}_{\lambda+(r^n);\kappa,n}(z;p)\quad (\lambda\in\PP_n, r\in\Z).
\end{equation*} 
It would be interesting to prove these relations since they would allow to prove the results in the present paper by the very same argument used in  Appendix~\ref{app:Jack} for $p=0$. 

\noindent {\bf 4.} As discussed in Remark~{\bf 3.} above, we believe that the functions $P^{(1/g)}_{\lambda;\kappa,,n}$ are an orthogonal system with respect to a natural scalar product $\langle\cdot,\cdot\rangle_n'$ exactly if $p$ is real and $\kappa$ purely imaginary, $\bar p=p$ and $\bar\kappa=-\kappa$ (we assume $g>0$). 
This is exactly the case when the equation in \eqref{Eq:Main} has a natural physical interpretation as non-stationary Schr\"odinger equation with  $\langle\cdot,\cdot\rangle_n'$ the standard $L^2$-scalar product  playing an important role in the quantum mechanical interpretation of this equation. 
To us, this is strong reason to think about the equation in \eqref{Eq:Main} as Schr\"odinger equation (rather than heat equation), even though the results in the present paper are for real $\kappa$ values.

\noindent {\bf 5.} 
All known CMS-like systems allow for deformations in the sense of \cite{CFV98}, 
and the non-stationary eCS equation is no exception: its deformation is given by 
\begin{equation*} 
\left(  \kappa\frac{\ii}{\pi} \frac{\partial}{\partial\tau}  + H_{n,\tilde{n}}(x,\tilde{x}) \ \right) \psi(x,\tilde{x}) = E\psi(x,\tilde{x})
\end{equation*} 
with the deformed eCS Hamiltonian 
\begin{multline*} 
H_{n,\tilde{n}}(x,\tilde{x})\equiv  -\sum_{j=1}^{n}\frac{\partial^2}{\partial x_j^2} + \sum_{j<j'}^{n} 2g(g-1)V(x_j- x_{j'}) \\
 +g\sum_{k=1}^{\tilde{n}}\frac{\partial^2}{\partial {\tilde x}_k^2} + \sum_{k<k'}^{\tilde{n}} \frac{2(g-1)}{g}V(\tilde{x}_k- \tilde{x}_{k'}) + \sum_{j=1}^n\sum_{k=1}^{\tilde{n}} 2(1-g) V(x_j- \tilde{x}_k)
\end{multline*} 
depending to two arbitrary variable numbers $n$ and $\tilde{n}$ (we suppress the $\tau$-dependence). 
This operator is known to define a quantum integrable system (this is proved in \cite{K05} for the case $\tilde{n}=1$, but the result is expected to be true for arbitrary $\tilde{n}$). 
Moreover, generalized kernel functions $K_{N,\tilde{N},M,\tilde{M}}(x,\tilde{x},y,\tilde{y})$ for this non-stationary deformed eCS equation are known \cite{Lan10}: 
they obey relations as in \eqref{gen} but for deformed eCS operators $H_{N,\tilde{N}}(x,\tilde{x})$ and $H_{M,\tilde{M}}(y,\tilde{y})$, and $\kappa=(N-M)g$ is replaced by $(N-M)g-\tilde{N}+\tilde{M}$. 
We expect that, using these generalized kernel functions, the construction in the present paper can be generalized to find solutions of the deformed non-stationary eCS equation.  
We leave this for future work. 

It is possible to generalize many results about the Jack polynomials in a natural way to the deformed case \cite{SV05,AL17,AHL18}. 
We thus believe that there is a natural generalization of the theory described in Remark~{\bf 1.} above to the deformed case as well. 
 
\noindent {\bf 6.} We believe that Theorem~\ref{thm} holds true for complex $g$ with real part $\Re(g)>0$. It would be interesting to prove this. 

\appendix
\section{Special functions}
\label{app:elliptic} 
For the convenience of the reader we collect the well-known definitions and properties of special functions that we use (we follow Whittaker \& Watson \cite{WW}).

The Weierstrass $\wp$-function with periods $(2\omega_1,2\omega_2)$ is defined as 
\begin{equation*}
\label{wpdef} 
 \wp (x|\omega_1,\omega_2)\equiv \frac{1}{x^2}+  \sum_{(m,n)\in \Z \times \Z \setminus \{ (0,0)\} } \left( \frac{1}{(x-\Omega_{m,n})^2}-\frac{1}{\Omega_{m,n}^2}\right)
\end{equation*} 
with $\Omega_{m,n}\equiv2m\omega_1 +2n\omega_2$ and $\Im(\omega_2/\omega_1)>0$. 
Moreover,
\begin{equation*}
\label{theta1def} 
\vartheta _1(x|\tau) \equiv 2\sum _{n=0}^{\infty } (-1)^{n} \ee^{\pi\ii \tau (n+1/2)^2}\sin\left((2n+1)x\right)
\end{equation*} 
with $\tau\equiv\omega_2/\omega_1$ defines the corresponding Jacobi theta function (we set $\omega_1=\pi$ in the main text). 

We quote two well-known results that we use \cite{WW}. 
First, 
\begin{equation*} 
\label{tildewp}
\wp(z|\omega_1,\omega_2) =  -\frac{\eta_1}{\omega_1} +\frac{\pi^2}{4\omega_1^2} \sum_{n\in\mathbb{Z}} \frac{1}{\sin^2(\frac{\pi (z+2n\omega_2)}{2\omega_1})} 
\end{equation*} 
with 
\begin{equation*} 
\label{eta1gen} 
\frac{\eta_1}{\omega_1} = \frac{\pi^2}{\omega_1^2}\left(\frac1{12} +\sum_{n=1}^\infty \frac1{2\sin^2(\frac{n\pi\omega_2}{\omega_1})} \right) , 
\end{equation*} 
and second, 
\begin{equation*} 
\label{theta1}
\vartheta_1(x|\tau) = 2p^{1/8} G \sin(x)\prod_{m=1}^\infty(1-2p^{m}\cos(2x)+p^{2m})
\end{equation*} 
with $G\equiv\prod_{m=1}^{\infty}(1-p^m)$ and $p=\ee^{2\pi\ii\tau}$ (note that our $p$ is $q^2$ in \cite{WW}).
These results imply the identities in \eqref{wp}--\eqref{eta1}, and that the function $\vartheta(x)$ in \eqref{vtet} equals $\vartheta_1(\half x|\tau)/p^{1/8} G $.

\section{Proof of Proposition~\ref{thm:AMOS}}
\label{app:Jack}
We give a concise proof of Proposition~\ref{thm:AMOS}. It is similar to the proof in \cite[Section~6]{AMOS95}, 
but we emphasize technical details allowing to use circles as integration contours even for non-integer $g$, 
different from \cite{AMOS95} (see Remark~\ref{rem:contour}). 

The identity in \eqref{GenerateJack} is in the sense of formal power series. 
However, we need information about convergence of the series on the RHS in \eqref{GenerateJack}.  
To get this we observe that the LHS in \eqref{GenerateJack} can be expanded in a Taylor series in the variables $z_j$ and $1/\xi_k$ using the binomial theorem repeatedly, 
and this series is absolutely and uniformly convergent in the region 
\begin{equation*} 
\label{region} 
|z_j/\xi_j|<\epsilon\quad \forall j,k, 
\end{equation*} 
for arbitrary $\epsilon$ in the range $0<\epsilon<1$.  
Thus, in such a region, we can re-sum this Taylor series so as to get the series on the RHS in \eqref{GenerateJack}, and the interchange of summation and integration below is justified. 

The second analytic fact we need is that the radius of the integration contours in the scalar product in \eqref{product} can be changed from $1$ to an arbitrary value $\rho>0$ without changing the result, i.e., 
\begin{equation} 
\label{product1} 
\langle f,g\rangle'_{n} \equiv \frac1{n!}\left(\prod_{j=1}^n\oint_{|z_j|=\rho}\frac{dz_j}{2\pi\ii z_j}\right) \left(\prod_{j\neq k}^n(1-z_j/z_k) \right)^g f(z)\overline{g(1/\bar{z})}\quad (f,g\in\Lambda_n)
\end{equation} 
is independent of $\rho>0$ (a concise proof is given at the end of this Appendix for the convenience of the reader, and more detailed discussion can be found in, e.g., \cite{AHL18}).
Moreover, for $g\in\Lambda_n$ with real coefficients, $\overline{g(1/\bar{z})}=g(1/z)$, and the Jack polynomials have this property.

To prove Proposition~\ref{thm:AMOS} we use the following identity implied by \eqref{GenerateJack}  and \eqref{shift}, 
\begin{equation*} 
\frac{\left(\prod_{j=1}^n z_j^r\right)\left(\prod_{k=1}^m \xi_k^{-r}\right)}{\prod_{j=1}^n\prod_{k=1}^m(1-z_j/\xi_k)^g} = 
\sum_{\nu\in\PP_m} \frac1{b_\nu(g)}P^{(1/g)}_{\nu+(r^n),n}(z)P^{(1/g)}_{\nu+(r^m),m}(1/\xi)\quad (m\leq n,\; r\in\Z)
\end{equation*} 
(the interpretation of this for $r<0$ is explained in Remark~\ref{rem:Jack}). 
For given Jack polynomial $P^{(1/g)}_{\mu,m}(\xi)$, we fix $\rho>0$ and assume $|z_j|<\rho$ for all $j$ to compute 
\begin{multline*} 
\left( \prod_{k=1}^m \oint_{|\xi_k|=\rho}\frac{d\xi_k}{2\pi\ii\xi_k}\right) \left(\prod_{j\neq k}^m(1-\xi_j/\xi_k) \right)^g
\frac{\left(\prod_{j=1}^n z_j^r\right)\left(\prod_{k=1}^m \xi_k^{-r}\right)}{\prod_{j=1}^n\prod_{k=1}^m(1-z_j/\xi_k)^g}\\
=  \sum_{\nu\in\PP_m} \frac1{b_\nu(g)}P^{(1/g)}_{\nu+(r^n),n}(z) \left( \prod_{k=1}^m \oint_{|\xi_k|=\rho}\frac{d\xi_k}{2\pi\ii\xi_k}\right)
\left(\prod_{j\neq k}^m(1-\xi_j/\xi_k) \right)^g P^{(1/g)}_{\mu,m}(\xi)
P^{(1/g)}_{\nu+(r^m),m}(1/\xi) \\
=  \sum_{\nu\in\PP_m} \frac1{b_\nu(g)}P^{(1/g)}_{\nu+(r^n),n}(z)m! \left\langle P^{(1/g)}_{\mu,m},P^{(1/g)}_{\nu+(r^m),m}\right\rangle'_{m}
\\=  \frac{m!\, \cN_{\mu,m}(g) }{b_{\mu-(r^m)}(g)}P^{(1/g)}_{\mu+(r^n)-(r^m),n}(z).
\end{multline*} 
In this computation we interchanged summation and integration (first equality), 
used \eqref{product1} and that the Jack polynomials have real coefficients (second equality), 
and inserted the orthogonality relations in \eqref{orth}       (third equality). 

It is important to note that, if $r>\mu_m$, then $ \left\langle P^{(1/g)}_{\mu,m},P^{(1/g)}_{\nu+(r^m),m}\right\rangle'_{m}=0$ for all $\nu\in\PP_m$, 
and thus we have to set $P^{(1/g)}_{\mu+(r^n)-(r^m),n}(z)\equiv 0$ in such a case. 
Note that 
\begin{equation*} 
\lambda\equiv \mu+(r^n)-(r^m) = (\mu_1,\ldots,\mu_m,r^{n-m}), 
\end{equation*} 
so this latter case exactly corresponds to $\lambda=(\lambda_1,\ldots,\lambda_n)$ {\em not} satisfying $\lambda_1\geq \cdots\geq \lambda_n$. 
Thus the result of this computation is an integral transform mapping a Jack polynomial $P^{(1/g)}_{\mu,m}$ to a Jack polynomial $P^{(1/g)}_{\lambda,n}$ 
provided that $\lambda$ satisfies $\lambda_1\geq \cdots\geq \lambda_n$, and the integral is 0 otherwise. 
Note that our derivation shows that the integral is only well-defined if $|z_j|<\rho$ for all $j$. 

We now apply this result to the family of integer vectors $\lambda^{(a)}$ in \eqref{lambdaa}, specializing to $(n,m,r)=(N_{a+1},N_a,r_{a+1})$ and $\mu=\lambda^{(a)}$. 
Since $\lambda^{(a)} + (r_{a+1}^{N_{a+1}})-(r_{a+1}^{N_a})=\lambda^{(a+1)}$ we get  
\begin{multline} 
\label{III}
 \frac{N_{a}!\, \cN_{\lambda^{(a)},N_{a}}(g) }{b_{\lambda^{(a)}-(r_{a+1}^{N_{a}})}(g)}P^{(1/g)}_{\lambda^{(a+1)},N_{a+1}}(z) 
 = 
\left(\prod_{j=1}^{N_{a+1}}z_j^{r_{a+1}}\right) \left( \prod_{k=1}^m \oint_{|\xi_k|=\rho}\frac{d\xi_k}{2\pi\ii\xi_k} \xi_k^{-r_{a+1}}\right) \\ \times 
\left( \frac{\prod_{j\neq k}^{N_{a}}(1-\xi_j/\xi_k)^g}{\prod_{j=1}^{N_{a+1}}\prod_{k=1}^{N_a}(1-z_j/\xi_k)^g}\right)
P^{(1/g)}_{\lambda^{(a)},N_a}(\xi)
\end{multline} 
for $|z_j|<\rho$ ($j=1,\ldots,n$). 
The result in Proposition~\ref{thm:AMOS} can be obtained from this by iteration, starting out with $a=1$ and 
\begin{equation*} 
P^{(1/g)}_{\lambda^{(1)},N_1}(z)=(z_1\cdots z_{N_1})^{r_1}\quad (N_1=s_1,\; \lambda^{(1)}=(r_1^{s_1})), 
\end{equation*} 
inserting the result into \eqref{III} for $a=2$, etc., up to $a=L-1$.

To keep this appendix self-contained we conclude with a proof that the RHS in \eqref{product1} is independent of $\rho>0$ and equal to the scalar product defined in \eqref{product}.
 
\begin{proof}[Proof of \eqref{product1}] We denote the RHS in \eqref{product1} as $(f,g)_\rho$. Note that $(f,g)_1$ is the RHS of \eqref{product}.

To prove the result it is enough to show it for monomial symmetric polynomials, 
$f=m_{\lambda}$ and $g=m_{\mu}$ with arbitrary $\lambda,\mu\in \PP_n$. Use polar coordinates $z_j=\rho\ee^{\ii x_j}$ to show that 
\begin{equation*} 
(m_{\lambda},m_{\mu})_\rho = \rho^{|\lambda|-|\mu|}(m_{\lambda},m_{\mu})_1=\rho^{|\lambda|-|\mu|}\langle m_{\lambda},m_{\mu}\rangle_n'\quad (\rho>0) 
\end{equation*} 
with $|\lambda|=\sum_j\lambda_j$ and similarly for $\mu$. Since the Euler operator $D_n\equiv \sum_{j=1}^n z_j\frac{\partial}{\partial z_j}$ is self-adjoint 
with respect to the scalar product $\langle\cdot,\cdot\rangle_n'$, and $D_nm_{\lambda}=|\lambda|m_{\lambda}$, $\langle m_{\lambda},m_{\mu}\rangle_n'=0$ unless $|\lambda|-|\mu|=0$. 
This implies the result. 
\end{proof} 

\section{Details of the proof in Section~\ref{sec:Proof}}
\label{app:Details}

\subsection{Analyticity}
\label{app:Analyticity} 
We prove that the integral transform in \eqref{precise2} defines an analytic function in the specified domain. 

We note that the integral in  \eqref{precise2} has the form 
\begin{equation} 
\label{Ffromf}
F(z;p) \equiv \left( \prod_{j=1}^M\oint_{|\xi_j|=\rho}\frac{d\xi_j}{2\pi\ii\xi_j}\right)\left( \prod_{j\neq k}^M\theta(\xi_j/\xi_k;p) \right)^g f(z,\xi;p) 
\end{equation} 
with the function 
\begin{equation*} 
f(z,\xi;p) = \frac{\left( \prod_{j=1}^N z_j^r\right)\left( \prod_{j=1}^{M} \xi_k^{-r'}\right)}{\prod_{j=1}^N\prod_{k=1}^M\theta\left(z_j/\xi_k;p \right)^g}\cP_M(\xi;p).
\end{equation*} 
Moreover, $f(z,\xi;p)$ is analytic in the following domain $\subset\C^{N+M+1}$, 
\begin{equation} 
\label{Domain} 
|p|<|z_j/\xi_k|<1\;\text{ and }\; \rho_0|p|<|\xi_k|<\rho_0\quad (j=1,\ldots,N,\, k=1,\ldots,M),\quad |p|<1
\end{equation} 
with $\rho_0>0$ specified in Lemma~\ref{lem:1} and $ \rho_0|p|<\rho<\rho_0$. Indeed, $1/\theta(z_j/\xi_k)^g$ is analytic in the region $|p|<|z_j/\xi_k|<1$ for all $j,k$, 
and the analyticity of the other factors is guaranteed by the assumptions on $\cP_M(\xi;p)$ and $r,r'$ in Lemma~\ref{lem:1}. 
We now {\bf claim}: {\em For every function $f(z,\xi;p)$ that is analytic in the region \eqref{Domain}, the integral  in \eqref{Ffromf} defines a function $F(z;p)$ that is analytic in the region \eqref{zjregion}.} 
This implies the analyticity result for the integral transform in \eqref{precise2} stated in Lemma~\ref{lem:1}. 

We are left to prove this claim. Since 
\begin{equation*} 
\left( \prod_{j\neq k}^M \theta(\xi_j/\xi_k;p)\right)^g = \prod_{j<k}^M\prod_{m\in\Z}\left[\left(1-p^{|m|}\frac{\xi_j}{\xi_k} \right)\left(1-p^{|m|}\frac{\xi_k}{\xi_j} \right) \right]^g
\end{equation*} 
and $|\xi_j/\xi_k|=1$ on the integration contours, it is clear that, for fixed $(z,p)$ in the region \eqref{zjregion}, 
the integrand in \eqref{Ffromf} is bounded on the integration contours, and thus $F(z;p)$ is well-defined. 
Moreover, $F(z;p)$ can be differentiated with respect to each variable $z_j$ $(j=1,\ldots,N)$ and $p$ separately by interchanging differentiation with integration, 
which is justified by the Leibniz  integral rule. 
Thus, $F(z;p)$ is analytic in each variable separately as long as the conditions in \eqref{zjregion}  are fulfilled and, by Hartogs's theorem, this implies that $F(z;p)$ is analytic in the region \eqref{zjregion}.

\subsection{Boundary terms}
\label{app:BT} 
We give details on how to prove the identities in \eqref{BT11} and \eqref{BT22}. 

Recalling the definition in \eqref{PN} and using the Leibniz integration rule we write the LHS in \eqref{BT11} as  an integral of total derivatives, 
\begin{equation*} 
-\ii\left(\prod_{j=1}^M\int_{-\pi-\ii\epsilon}^{\pi-\ii\epsilon} \frac{dy_j}{2\pi}\right)  \sum_{k=1}^M \frac{\partial}{\partial y_k} K_{NM}(x,y)\ee^{-\ii Q'|y|}\psi_M(y). 
\end{equation*}
Inserting definitions we find that this is proportional to (we ignore irrelevant $y$-independent factors) 
\begin{equation*} 
 \left(\prod_{j=1}^M\int_{-\pi-\ii\epsilon}^{\pi-\ii\epsilon} \frac{dy_j}{2\pi}\right) \sum_{k=1}^M  \frac{\partial}{\partial y_k}F(y),\quad F(y)\equiv 
\frac{\Psi_M(y)^2 }{\prod_{j=1}^N\prod_{k=1}^M\theta(\ee^{-\ii y_k}z_j/\rho)^g}\ee^{-\ii r'|y|}\cP_M(\rho\ee^{\ii y})
\end{equation*} 
where $\rho\ee^{\ii y}$ is short for $(\rho\ee^{\ii y_1},\ldots,\rho\ee^{\ii y_M})$. 
All factors but $\Psi_M(y)^2$ in the function $F(y)$ are obviously $C^\infty$ and $2\pi$-periodic in each variable $y_j$ for $y\in[-\pi,\pi]^M$ 
(recall our analyticity assumptions on $\cP_M$, and that $r'\in\Z$ and $|p|<|z_j|/\rho<1$). 
The only potentially problematic factor that  could spoil that the total derivatives of $F(y)$ integrate to 0 is $\Psi_M(y)^2$ for non-integer $g$. 

Using definitions we can write 
\begin{equation*} 
\Psi_M(y)^2 = \left( \prod_{j<k}^M\sin^2\left(\frac{y_j-y_k}{2}\right)\right)^g \left(\prod_{j<k}^N\prod_{m=1}^\infty(1-2p^{m}\cos(y_j-y_k)+p^{2m}) \right)^{2g}, 
\end{equation*} 
which makes manifest that, actually, the only potentially problematic factor is 
\begin{equation} 
\Psi^{(0)}_M(y)^2 \equiv  \left( \prod_{j<k}^M\sin^2\left(\frac{y_j-y_k}{2}\right)\right)^g 
\end{equation} 
 (recall that $|p|<1$). 
 We now show that $\Psi^{(0)}_M(y)^2$ is $C^1$ and $2\pi$-periodic in each variable $y_j$ for $y\in[-\pi,\pi]^M$. 
 This implies the same properties for the function $F(y)$, which proves \eqref{BT11}. 

The $2\pi$-periodicity of $\Psi^{(0)}_M(y)^2$ is obvious. 
To prove the $C^1$-property we compute the partial derivative $\frac{\partial}{\partial y_j}\Psi^{(0)}_M(y)^2$ and find that it is a finite sum of terms proportional to
\begin{equation*} 
\cot\left(\frac{y_j-y_k}{2}\right)\Psi^{(0)}_M(y)^2\quad (k\neq j). 
\end{equation*} 
This behaves like $(y_j-y_k)^{-1+2g}$ as $y_j\to y_k$, and it is obviously is continuous everywhere else. 
Since we assume $g>1/2$, this proves that $\Psi^{(0)}_M(y)^2$ is $C^1$. 
This concludes the proof of \eqref{BT11}.

The proof of \eqref{BT22} is similar to the proof of \eqref{BT11},  and we are therefore sketchy. 
Inserting definitions we find that the LHS is an integral of total derivatives, 
\begin{multline*} 
\label{BT222} 
\left(\prod_{j=1}^M\int_{-\pi-\ii\epsilon}^{\pi-\ii\epsilon} \frac{dy_j}{2\pi}\right)  \sum_{k=1}^M \frac{\partial}{\partial y_k}\biggl(- \left\{\frac{\partial}{\partial y_k} K_{NM}(x,y)\right\} \ee^{-\ii Q'|y|}\psi_M(y) \\
+  K_{NM}(x,y) \frac{\partial}{\partial y_k} \ee^{-\ii Q'|y|}\psi_M(y) \biggr) . 
\end{multline*}  
Computing the terms one finds again that the only potentially problematic term comes from the factor $\Psi^{(0)}_M(y)^2$ in the integrand. 
As in the proof of  \eqref{BT11} above one checks that this factor, actually, is not problematic, either. 
The only difference to \eqref{BT11} is that we now have second order derivative terms $\frac{\partial^2}{\partial y_j^2}\Psi^{(0)}_M(y)^2$ which give rise to terms (we ignore less singular terms)
\begin{equation*} 
\frac{1}{\sin^2\left(\frac{y_j-y_k}{2}\right)}
\Psi^{(0)}_M(y)^2\quad (k\neq j), 
\end{equation*} 
which behaves like $(y_j-y_k)^{-2+2g}$ as $y_j\to y_k$ (continuity otherwise is again obvious). 
While these terms are singular for $1/2\leq g<1$ (they are regular for $g\geq 1$), the singularities are integrable for $g>1/2$, and this is enough to conclude that \eqref{BT22} holds true. 

\subsection{Eigenvalue computations}
\label{app:computation}
We give some details of computations that verify  \eqref{ppNa} and \eqref{ENa}.

We start with \eqref{ppNa}. We insert \eqref{pthm} for the cases $n=N_a$ and $n=N_{a-1}$  to compute 
\begin{equation*} 
d_{N_a} -d_{N_{a-1}}= \sum_{j=N_{a-1}+1}^{N_a}\lambda_j = k r_a
\end{equation*} 
using $\lambda_j=r_a$ for $N_{a-1}+1\leq j\leq N_a$ and $N_a=N_{a-1}+k$ {(recall that we assume that $\lambda$ is as in \eqref{lambdagen})}. 
From \eqref{iterate} it follows that this is equal to $N_a Q_a-N_{a-1}Q_a'$, which proves \eqref{ppNa}. 

To prove \eqref{ENa}  we use \eqref{Ethm} to compute, in a similar manner, 
\begin{multline*} 
E_{N_a} -E_{N_{a-1}}= 
\sum_{j=N_{a-1}+1}^{N_a}\left( r_a + \half g ( N_{a} +1 - 2 j) \right)^{2} + \sum_{j=1}^{N_{a-1}}\biggl( gk \left( \lambda_j + \half g(N_{a-1}+1-2j)\right)  \\
+\frac{1}{4}(gk)^2\biggr) +g^2\left[N_{a}(N_{a}-1) -N_{a-1}(N_{a-1}-1)\right]\left(\frac{\eta_1}{\pi}-\frac1{12} \right)  \\
= k r_a^2 - g k r_aN_{a-1} 
+ \frac1{12} k g^2\left( 3N_{a-1}^2 +3N_{a-1}k + k^2-1\right)\\  
+ gk d_{N_{a-1}}+ g^2k(2N_{a-1}+k-1)\left(\frac{\eta_1}{\pi}-\frac1{12} \right) , 
\end{multline*} 
recalling \eqref{pthm} in the last step. Using  \eqref{pthm},  \eqref{cNM} and \eqref{iterate} one can check that this equals 
$2(Q_a-Q'_a)d_{N_{a-1}} +N_{a-1}(Q'_a)^2-2N_{a-1}Q_aQ_a' +  N_aQ_a^2 + c_{N_aN_{a-1}}$, which proves \eqref{ENa}.

\end{document}